\documentclass{article}
\usepackage{graphicx} % Required for inserting images
\usepackage{hyperref}
\usepackage{url}
\usepackage{amsmath}
\usepackage{amssymb}
\usepackage{geometry}
\usepackage{multirow}
\usepackage{todonotes}
\usepackage{booktabs}
\usepackage{hyperref} 
\usepackage{url}

\newcommand{\Eidolon}{{\tt Eidolon}}

\usepackage{tikz}
\usepackage{array}
\usepackage{stmaryrd} 

\usepackage{algorithm}
\usepackage{algpseudocode}

\usepackage{listings}
\usepackage{xcolor}

\newcommand{\email}[1]{\texttt{#1}}

\definecolor{codegreen}{rgb}{0,0.6,0}
\definecolor{codegray}{rgb}{0.5,0.5,0.5}
\definecolor{codepurple}{rgb}{0.58,0,0.82}
\definecolor{backcolour}{rgb}{0.95,0.95,0.92}

\lstdefinestyle{mystyle}{
    backgroundcolor=\color{backcolour},
    commentstyle=\color{codegreen},
    keywordstyle=\color{magenta},
    numberstyle=\tiny\color{codegray},
    stringstyle=\color{codepurple},
    basicstyle=\ttfamily\footnotesize,
    breakatwhitespace=false,
    breaklines=true,
    captionpos=b,
    keepspaces=true,
    numbers=left,
    numbersep=5pt,
    showspaces=false,
    showstringspaces=false,
    showtabs=false,
    tabsize=2,
    frame=single,
    rulecolor=\color{gray!30}
}

\usepackage{amsthm}   % For theorem environments
\usepackage[utf8]{inputenc}
\usepackage[T1]{fontenc}
\usepackage{lmodern}

\theoremstyle{definition}
\newtheorem{definition}{Definition}[section]

\theoremstyle{plain}
\newtheorem{proposition}{Proposition}[section]

\newtheorem{lemma}{Lemma}[section]

\theoremstyle{remark}

\usepackage{float}

\title{\Eidolon{}: A Post-Quantum Signature Scheme Based on k-Colorability in the Age of Graph Neural Networks}

\author{
Asmaa Cherkaoui\thanks{Laboratory of Mathematical Analysis, Algebra and Applications (LAM2A), Faculty of Sciences Ain Chock (FSAC), University Hassan II, Casablanca, Morocco. \email{esma1maysan@gmail.com}} \\
Ram\'on Flores\thanks{Department of Geometry and Topology, Faculty of Mathematics, University of Seville, Seville, Spain. 
\email{ramonjflores@us.es}}\\
Delaram Kahrobaei\thanks{Departments of Computer Science and Mathematics, Queens College, City University of New York, USA;
PhD Program in Mathematics, and Initiative for the Theoretical Sciences, Graduate Center, City University of New York, USA;
Department of Computer Science and Engineering, Tandon School of Engineering, New York 
University, USA; Department of Computer Science, University of York, United Kingdom
\email{delaram.kahrobaei@qc.cuny.edu}} \\
Richard C. Wilson\thanks{Department of Computer Science, University of York, United Kingdom
\email{richard.wilson@york.ac.uk}} 
}

\begin{document}

\maketitle

\begin{abstract}
We propose \Eidolon{}, a post-quantum signature scheme grounded in the NP-complete $k$-colorability problem. Our construction generalizes the Goldreich–Micali–Wigderson zero-knowledge protocol to arbitrary $k \geq 3$, applies the Fiat–Shamir transform, and uses Merkle-tree commitments to compress signatures from $O(tn)$ to $O(t \log n)$. We generate instances by planting a coloring while aiming to preserve the statistical profile of random graphs. We present an empirical security analysis of such a scheme against both classical solvers (ILP, DSatur) and a custom graph neural network (GNN) attacker. Experiments show that for $n \geq 60$, neither approach is able to recover a valid coloring matching the planted solution, suggesting that well-engineered $k$-coloring instances can resist the considered classical and learning-based cryptanalytic approaches. These experiments indicate that the constructed instances resist the attacks considered in our evaluation.
\end{abstract}

\noindent\textbf{Keywords:} Post-quantum cryptography, $k$-colorability, NP-hard signatures, zero-knowledge proofs, fiat–shamir transform, graph neural networks.

\section{Introduction}

With the rising threat of quantum computers to traditional cryptography, recent attention has focused on post-quantum cryptographic methods. Based on the current belief that there is no quantum speed-up for NP-complete problems, these problems are potentially a rich source of potential cryptosystems. In particular, graph theory contains several NP-complete problems, including homomorphism, \(k\)-colorability, and Hamiltonicity. In this paper, we study the application of \(k\)-colorability to a post-quantum signature system.

When considering an NP-hard problem as the basis for a cryptosystem, care must be taken in the concrete realization of a particular instance. This is because many instances can be solved in acceptable time using heuristic algorithms. Indeed, a number of heuristics are available for graph coloring and can quickly solve the problem for many graphs.

In this regard, machine learning has emerged as a particular threat because of its ability to extract powerful heuristics directly from data. In this paper, we explore the issue of instance hardness for \(k\)-colorability, from both a theoretical and empirical perspective, to demonstrate the security of the system. In \cite{SBK}, the authors demonstrate how graph neural networks can be used to solve combinatorial optimization problems. Their approach is broadly applicable to canonical NP-hard problems in the form of quadratic unconstrained binary optimization problems, such as maximum cut, minimum vertex cover, maximum independent set, as well as Ising spin glasses and higher-order generalizations thereof in the form of polynomial unconstrained binary optimization problems.

The rapid progress in quantum computing has intensified the search for cryptographic primitives that remain secure in a post-quantum world. While lattice- and code-based schemes currently dominate standardization efforts, combinatorial problems offer an alternative foundation rooted in computational complexity theory. The use of \(k\)-colorability in cryptography is limited by the gap between worst-case hardness and typical-instance behavior: many graph instances that are hard in the worst-case sense can still be handled efficiently by heuristic algorithms such as DSatur~\cite{Brelaz}, and more recently by data-driven methods leveraging machine learning.

Recent work by Schuetz, Brubaker, and Katzgraber~\cite{SBK} demonstrates that graph neural networks (GNNs) can effectively approximate solutions to canonical NP-hard problems, including graph coloring, by learning implicit heuristics from data. This motivates examining whether such methods affect the security of schemes based on combinatorial hardness. To date, to the best of our knowledge, no signature scheme based on \(k\)-colorability has been proposed that simultaneously offers a complete construction, incorporates instance generation that embeds a secret coloring while aiming to maintain statistical indistinguishability from a random graph, and undergoes empirical validation against both classical and learning-based attacks.

In this paper, we present \Eidolon{}, a post-quantum digital signature scheme derived from the \(k\)-colorability problem by generalizing the zero-knowledge identification protocol of Goldreich, Micali, and Wigderson (GMW)~\cite{GMW}, originally demonstrated for 3-colorability, to arbitrary \(k \geq 3\), and then applying the Fiat--Shamir transform. Our design features a planted coloring mechanism, inspired by the ``quiet solution'' framework of Krzakala and Zdeborová~\cite{KZ}, which embeds a secret \(k\)-coloring into a \(k\)-partite random graph with calibrated edge density, ensuring witness existence while aiming to maintain statistical indistinguishability from an Erd\H{o}s--R\'enyi random graph. To reduce the signature size obtained from a direct Fiat--Shamir transformation of the underlying zero-knowledge protocol, we use Merkle-tree vector commitments, compressing the per-round vertex commitments into a single root and thereby reducing the overall signature size from \(O(tn)\) to \(O(t\log n)\), where \(t\) denotes the number of Fiat--Shamir rounds and \(n = |V|\).

We evaluate the scheme empirically using two attack strategies: the classical DSatur heuristic and a custom-designed graph neural network (GNN). Our GNN is inspired by the framework of Schuetz, Brubaker, and Katzgraber~\cite{SBK}. Both attackers are tested on the same class of \(k\)-partite random graphs. Our experiments show that, for graphs with \(n \geq 60\) vertices and our chosen density parameters, neither DSatur nor our GNN is able to recover a coloring matching the planted solution, providing empirical evidence, within the tested regime, that the hardness of our \(k\)-colorability instances withstands the considered classical and learning-based cryptanalytic approaches.

In summary, our contributions are threefold: (i) we design \Eidolon{}, a Fiat--Shamir signature scheme based on \(k\)-colorability with a statistically hidden planted coloring; (ii) we reduce its asymptotic signature size via Merkle-tree vector commitments; and (iii) we provide an empirical hardness study against both classical heuristics and GNN-based attacks. The rest of the paper is organized as follows: Section~\ref{InstanceSelection} reviews graph coloring and instance generation; Section~\ref{sec:scheme} details our signature construction; and Section~\ref{sec:attacks} describes our attack models and experimental setup.

\section{Graph instance generation}
\label{InstanceSelection}

Let $G=(V,E)$ be a graph with vertices $V$ and edges $E$. A $k$-coloring of a graph is an assignment of one of $k$ colors to each vertex so that no two adjacent vertices share the same color. It is known that determining whether a graph has a $k$-coloring is NP-complete for $k \geq 3$. The minimum number of colors needed is the \emph{chromatic number} $\chi(G)$, and computing it is NP-hard.

However, worst-case hardness alone is insufficient for cryptography: many concrete instances are easy to solve. For example, if $k = |V|$, a trivial coloring exists; complete graphs have $\chi(G) = |V|$; and heuristic algorithms like DSatur often find valid colorings quickly on structured or sparse graphs. This creates a constraint for signature schemes based on \(k\)-colorability: the prover must know a valid \(k\)-coloring, which serves as the secret key, while the public graph should not reveal that coloring to the verifier or to an adversary.

Thus, we must \emph{construct} instances where:
\begin{itemize}
    \item a valid $k$-coloring is known by design to the prover,
    \item the graph appears statistically indistinguishable from a random hard instance,
    \item and recovering the coloring remains infeasible for classical and ML-based attackers.
\end{itemize}
This ensures the protocol is both \emph{correct} (the prover can always respond) and \emph{secure} (the secret remains hidden).

An (Erd\"os-R\'enyi) random graph is a graph $G_R(n,p)$ on $n$ vertices, where each vertex pair is joined uniformly at random with probability $p$. A number of results are available on the hardness of coloring random graphs. In \cite{Garey} it is shown that coloring a graph with less than $2\chi(G)-\delta$ is NP-hard. For fixed $p$, almost every random graph has a chromatic number \cite{Bollobas}
\begin{equation}
    \chi(G_R(n,p))=\frac{n}{r}(1+\epsilon)
\end{equation}
with $0\leq \epsilon\leq 3\log \log n/\log n$ and
\begin{equation}
r=2\log_d n-\log_d\log_d n+2\log_d(e/2)+1.
\end{equation}
These results hold in the limit of large $n$.
Furthermore, we observe that very sparse and very dense graphs are easy to color, and so we select $p=1/2$. In this case, it has been shown that  $\chi(G_R(n,p))\sim \log\left(1/(1-p)\right)/2{\log np}$ \cite{Mcdiarmid}.
However, these results are asymptotic and may not hold for the small graphs under consideration here. We reserve selection of $k$ for the planted coloring until our empirical analysis, noting only that it should be the same size as the expected chromatic number of the equivalent random graph.

\subsection{Planted $k$-colorable graph construction}
\label{PlantedSolution}
In order to operate the digital signature scheme above, it is necessary to be able to construct a graph and $k$-coloring in polynomial time, where the coloring is difficult to discover. We use the natural algorithm for planting a known coloring in a random graph \cite{KZ}. In particular, we begin by selecting a vertex set $V$, $|V|=n$ (the problem size) and $k\approx\chi[G_R(n,p)]$. We then follow the following steps.
\begin{itemize}
    \item $V$ is partitioned into $k$ sets such that $n_i=|V_i|, n_i\in \left\{ \lfloor n/k \rfloor, \lceil n/k \rceil\right\}$ and $\sum_i n_i=n$
    \item We iterate through all pairs of vertices $(u,v)$
    \item Let $P_u,P_v$ be the partitions of the vertices. We join the vertices with probability $p$ if $P_u\neq P_v$ and probability zero otherwise.
\end{itemize}

The resulting graph may be colored with $k$ colors simply by assigning one color to each partition. In \cite{KZ}, the authors demonstrate that this is a \emph{quiet solution} which is hidden in the random graph in the sense that the properties of the graph are not much altered by the existence of the known extra solution. In particular, they conjecture that the hardness of the problem is the same as for the original graph.

In general we work with random graphs such that the probability of an edge between two given vertices is a fixed number $p\in [0,1]$, in such a way that if the graph has $n$ vertices, the expected final number of edges is $p {n\choose 2}$. We consider here a multipartite graph $\Gamma$ with $n$ vertices and a partition in the set of vertices given by $V=V_1\cup\ldots V_k$. 
Given a number $s\in [0,1]$, we compute here the probability $p$ of existence of an edge in $\Gamma$ such that the expected number of edges in the graph is $s {n\choose 2}$.

Recall that a complete graph in $n_j$ vertices has ${n_j\choose 2}$ edges. Hence, the number of forbidden edges in $\Gamma$ is $S=\sum_{j=1}^k {n_j\choose 2}$. Then, the maximum number of edges of the graph $\Gamma$ is ${n\choose 2}-S$, and given a probability $p$ of existence of an edge, the expected number of edges in the graph is $p ({n\choose 2}-S)$. We have imposed that $$p ({n\choose 2}-S)=s {n\choose 2},$$ and hence $$p=\frac{s {n\choose 2}}{{n\choose 2}-S}.$$

Observe that $s {n\choose 2}$ is always bounded in the graph $\Gamma$ by ${n\choose 2}-S$, and hence we always obtain $p\leq 1$. 
Edges in the graph are selected with probability $p$. We discuss the specific values of $p$ and $k$ in Section \ref{sec:attacks} on our security analysis.

%For example, if $\Gamma$ is $(8,6)$-bipartite and we expect 40 edges in the outcome, then $s=0.439$ and $p=0.832$. 

\section{Protocol description} 
\label{sec:scheme}
We describe the protocol underlying \Eidolon{}, based on graph \(k\)-colorability. Let
\[
f:\{1,\dots,k\}\times\{0,1\}^r \longrightarrow \{0,1\}^s
\]
be a statistically hiding and computationally binding commitment function, as in \cite{GMW}. One round of the protocol is as follows. The protocol is then repeated independently \(T\) times.

\begin{itemize}
  \item \textbf{Coloring.} The prover holds a valid $k$-coloring
  $\phi:V\to\{1,\dots,k\}$.

  \item \textbf{Permutation.} The prover samples a fresh random permutation
  $\pi\in S_k$ to mask color labels.

  \item \textbf{Commitment.} For each vertex $i\in V$, the prover samples fresh
  randomness $r_i$ and commits to the permuted color
  \[
  c_i \;=\; f\big(\pi(\phi(i)),\,r_i\big).
  \]
  The prover sends all commitments $\{c_i\}_{i\in V}$ to the verifier.

  \item \textbf{Challenge.} The verifier chooses a uniform random edge
  $(u,v)\in E$ and requests openings for its endpoints.

  \item \textbf{Response.} The prover reveals
  $\big(\pi(\phi(u)),r_u\big)$ and $\big(\pi(\phi(v)),r_v\big)$.

  \item \textbf{Verification.} The verifier checks
\[
f\big(\pi(\phi(u)),r_u\big)\stackrel{?}{=}c_u,\qquad
f\big(\pi(\phi(v)),r_v\big)\stackrel{?}{=}c_v,
\]
and verifies that
\[
(u,v)\in E,\qquad
\pi(\phi(u)),\pi(\phi(v))\in\{1,\dots,k\},\qquad
\pi(\phi(u))\neq \pi(\phi(v)).
\]
The round is accepted if all checks pass.
\end{itemize}

\noindent\textbf{Zero-knowledge intuition.}
In each round, the prover uses a fresh permutation \(\pi\) and fresh commitment randomness. The verifier therefore learns only that the challenged edge joins two differently colored vertices under the permuted coloring.

\subsection{Soundness} 
To analyze the soundness of the zero-knowledge proof protocol for \( k \)-colorability, consider the case where the graph is \emph{not} \( k \)-colorable. In such a scenario, no matter how the prover attempts to simulate a valid coloring, there will inevitably be a set of \emph{violating edges}--edges whose endpoints receive the same color under any attempted coloring. Let \( t \) denote the number of these bad edges, and let \( m \) be the total number of edges in the graph. Since the verifier selects one edge uniformly at random in each round, the probability that a cheating prover is caught in a single round is at least \( \frac{t}{m} \), while the probability of escaping detection is at most \( 1 - \frac{t}{m} \).

Although this is the general case, we often assume a worst-case scenario where \( t \geq 1 \) but unknown. In this case, we conservatively lower-bound the detection probability per round by \( \frac{1}{m} \). The protocol is repeated independently for \( m^2 \) rounds to drive the cheating probability down. Therefore, the probability that a cheating prover escapes detection across all \( m^2 \) rounds is:
\[
\left(1 - \frac{1}{m} \right)^{m^2} = \left[ \left(1 - \frac{1}{m} \right)^m \right]^m.
\]
It is well known that:
\[
\lim_{m \to \infty} \left(1 - \frac{1}{m} \right)^m = \frac{1}{e},
\]
so it follows that:
\[
\left(1 - \frac{1}{m} \right)^{m^2} \approx \left( \frac{1}{e} \right)^m = e^{-m}.
\]
This quantity becomes \emph{negligibly small} as \( m \) increases, for instance, \( e^{-40} < 2^{-55} \). Thus, the verifier's probability of accepting a false claim is exponentially small in the number of edges.

The number of rounds is taken to depend on the number of edges \(m\), rather than on the number of vertices \(n\) or the number of colors \(k\), because the verifier’s challenges are edge-based and violations are detected on edges.

\subsection{Identification and signature construction}

We use a GMW-type identification protocol for graph \(k\)-coloring, following the classical approach of Goldreich, Micali, and Wigderson~\cite{GMW}. The prover commits to a randomly permuted coloring of the public graph and, upon receiving a challenge edge, opens only the two commitments corresponding to its endpoints. The verifier then checks that the revealed colors are valid, distinct, and consistent with the commitments.

\begin{center}
\begin{tabular}{>{\raggedright}p{7cm} >{\raggedright\arraybackslash}p{7cm}}
\textbf{PROVER} & \textbf{VERIFIER} \\
\hline
Hold a valid coloring $\phi:V\to\{1,\dots,k\}$ & \\
Sample a fresh permutation $\pi\in S_k$ and fresh $r_v$ for all $v\in V$ & \\
Compute $c_v \,=\, f\!\big(\pi(\phi(v)),\, r_v\big)$ for all $v\in V$ & \\
Set $C \,=\, \{c_v\}_{v\in V}$ & \\
\multicolumn{2}{c}{$\xrightarrow{\hspace{1em}C\hspace{1em}}$} \\
& Sample a uniform random edge $(u,v)\in E$ \\
\multicolumn{2}{c}{$\xleftarrow{\hspace{1em}(u,v)\hspace{1em}}$} \\
Reveal $(\pi(\phi(u)), r_u)$ and $(\pi(\phi(v)), r_v)$ & \\
\multicolumn{2}{c}{$\xrightarrow{\hspace{1em}(\pi(\phi(u)), r_u),\; (\pi(\phi(v)), r_v)\hspace{1em}}$} \\
& Accept if \\
& \quad $(u,v)\in E$ \\
& \quad $f(\pi(\phi(u)), r_u) = c_u$ \quad and \quad $f(\pi(\phi(v)), r_v) = c_v$ \\
& \quad $\pi(\phi(u)),\pi(\phi(v))\in\{1,\dots,k\}$ \quad and \quad $\pi(\phi(u)) \ne \pi(\phi(v))$ \\
\hline
\end{tabular}
\vspace{0.5em}
Table 1: Graph $k$-Coloring Identification Scheme (one round; repeat $T$ times with fresh $\pi$ and randomness)
\end{center}

\subsection{Fiat--Shamir signature construction}

To obtain a non-interactive signature scheme, we apply the Fiat--Shamir transform~\cite{FIAT} to the identification protocol in the random-oracle model. The signer first computes the round commitments using fresh permutations and fresh randomness, then evaluates
\[
h=\mathcal H(\mathrm{Encode}(G,k,C^{(0)},\dots,C^{(t-1)},M))
\]
on the public data, the commitments, and the message \(M\), where \(\mathrm{Encode}(\cdot)\) is a canonical encoding procedure. A public hash-to-edges parser deterministically derives the challenge edges from \(h\), and the signer includes the openings corresponding to the endpoints of those edges. Verification recomputes the same hash from \((G,k)\), the commitments, and \(M\), derives the same challenge edges, and checks the openings together with color distinctness.

\paragraph{Challenge derivation.}
The challenge edges are derived deterministically from the hash value \(h\). More precisely, we sample \(t\) edge indices with replacement so as to match the independent public-coin challenges of the identification protocol. To avoid modulo bias, the mapping from hash output to edge indices is implemented using domain-separated rejection sampling.

\begin{algorithm}[H]
\caption{\textsf{HashToEdges}$\,(h, t, E)$}
\begin{algorithmic}[1]
\State $m\gets |E|$;\quad $b\gets \lambda$ \Comment{$\lambda$-bit blocks from $\mathcal H$}
\For{$i=0$ to $t-1$}
  \State $j\gets 0$
  \Repeat
    \State $B \gets \mathcal H(\texttt{``EdgeDerive-v1''}\ \Vert\ h\ \Vert\langle i\rangle_{32}\Vert\langle j\rangle_{32})$
    \State $x \gets \mathrm{int}(B)$ \Comment{interpret $B$ as a big-endian integer in $[0,2^{\lambda}\!-\!1]$}
    \State $M \gets \big\lfloor 2^{\lambda}/m \big\rfloor\cdot m$
    \State $j \gets j+1$
  \Until{$x < M$} \Comment{rejection sampling removes modulo bias}
  \State $\mathrm{idx}_i \gets x \bmod m$;\quad $e_i \gets E[\mathrm{idx}_i]$
\EndFor
\State \Return $(e_0,\dots,e_{t-1})$
\end{algorithmic}
\end{algorithm}

\paragraph{Public inputs and secret key.}
Public parameters are $(G=(V,E),k,f,\mathcal H,\lambda)$, where $V=\{1,\dots,n\}$,
$E\subseteq \{\{u,v\}\mid 1\le u<v\le n\}$, $f$ is a statistically hiding, computationally binding
commitment, and $\mathcal H:\{0,1\}^*\to\{0,1\}^\lambda$ is modeled as a random oracle.
The \emph{secret key} is a valid $k$-coloring $\phi:V\to\{1,\dots,k\}$.

\paragraph{Canonical encoding with domain separation.}

We follow standard practice for domain separation \cite{SP800-185,RFC9380} and canonical serialization
\cite{RFC6979,RFC5869}, using a fixed domain-separation tag $\mathsf{TAG}=\texttt{``FS-GkColor-v1''}$ and
a canonical serializer $\mathrm{Encode}(\cdot)$:

\[
\small
\begin{aligned}
\mathrm{Encode}(G,k,X_0,\dots,X_{t-1},M)
&=\ \langle \mathsf{TAG}\rangle\ \Vert\
   \langle n\rangle_{64}\ \Vert\ \langle k\rangle_{32}\ \Vert\
   \langle m\rangle_{64}\ \Vert\ \mathrm{Edges}(E) \\
&\quad\Vert\ \langle t\rangle_{32}\ \Vert\
   X_0\Vert\cdots\Vert X_{t-1}\ \Vert\
   \langle |M|\rangle_{64}\ \Vert\ M.
\end{aligned}
\]
Here $\langle\cdot\rangle_{b}$ is a $b$-bit big-endian length/value encoding, and $\Vert$ is concatenation.
We fix a vertex order $1<\cdots<n$ and encode edges as
\[
\mathrm{Edges}(E)=\langle m\rangle_{64}\ \Vert\
\big\lVert_{\{u,v\}\in E^{\uparrow}} \big(\langle u\rangle_{\lceil\log_2 n\rceil}\ \Vert\ \langle v\rangle_{\lceil\log_2 n\rceil}\big),
\]
where $E^{\uparrow}$ lists edges with $u<v$ in lexicographic order. This canonicalization ensures all parties
hash identical byte strings.

\subsubsection{\Eidolon{} signature scheme}

\Eidolon{} uses as public parameters a graph \(G=(V,E)\), an integer \(k\in\mathbb N\), a statistically hiding and computationally binding commitment
\[
f:\{1,\dots,k\}\times\{0,1\}^r\to\{0,1\}^s,
\]
and a hash function \(\mathcal H:\{0,1\}^*\to\{0,1\}^\lambda\), modeled as a random oracle. The secret key is a valid \(k\)-coloring \(\phi:V\to\{1,\dots,k\}\). Let \(m=|E|\) and set \(t=m^2\).

To sign a message \(M\in\{0,1\}^*\), the signer samples, for each \(i\in\{0,\dots,t-1\}\), a fresh permutation \(\pi_i\in S_k\) and, for every vertex \(v\in V\), fresh randomness \(r_v^{(i)}\), and computes
\[
c_v^{(i)} = f\big(\pi_i(\phi(v)),\, r_v^{(i)}\big).
\]
Let \(C^{(i)}=(c_v^{(i)})_{v\in V}\). Then compute
\[
h = \mathcal H\!\big(\mathrm{Encode}(G,k,C^{(0)},\dots,C^{(t-1)},M)\big),
\]
where \(\mathrm{Encode}(\cdot)\) is a fixed canonical encoding with domain separation. The value \(h\) is parsed deterministically into a sequence of challenge edges \((e_0,\dots,e_{t-1})\), where \(e_i=(u_i,v_i)\in E\). For each \(i\), the signer includes the openings
\[
\big(\pi_i(\phi(u_i)),\, r_{u_i}^{(i)}\big)
\qquad\text{and}\qquad
\big(\pi_i(\phi(v_i)),\, r_{v_i}^{(i)}\big).
\]
The signature is
\[
\sigma = \Big(\,\{C^{(i)}\}_{i=0}^{t-1},\; \{\pi_i(\phi(u_i)),r_{u_i}^{(i)},\pi_i(\phi(v_i)),r_{v_i}^{(i)}\}_{i=0}^{t-1}\,\Big).
\]
\begin{algorithm}[H]
\caption{\textsc{Eidolon.Sign}$(G,k,\phi,M)$}
\begin{algorithmic}[1]
\State \textbf{Input:} message \(M\), private coloring \(\phi:V\to\{1,\dots,k\}\)
\For{$i=0,\dots,t-1$}
    \State choose a fresh random permutation \(\pi_i\in S_k\)
    \For{each \(v\in V\)}
        \State choose fresh randomness \(r_v^{(i)}\) and compute
        \[
        c_v^{(i)} \gets f(\pi_i(\phi(v)),\, r_v^{(i)}).
        \]
    \EndFor
    \State set \(C^{(i)} \gets (c_v^{(i)})_{v\in V}\)
\EndFor
\State Compute
\[
h \gets \mathcal H\big(\mathrm{Encode}(G,k,C^{(0)},\dots,C^{(t-1)},M)\big).
\]
\State Parse \(h\) deterministically to obtain challenges \(e_0,\dots,e_{t-1}\), where each \(e_i\) encodes an edge \((u_i,v_i)\in E\)
\For{$i=0,\dots,t-1$}
    \State set
    \[
    \mathit{open}_i \gets \big(\pi_i(\phi(u_i)),\, r_{u_i}^{(i)},\; \pi_i(\phi(v_i)),\, r_{v_i}^{(i)}\big)
    \]
\EndFor
\State Output signature
\[
\sigma \gets \Big( \{C^{(i)}\}_{i=0}^{t-1},\; \{\mathit{open}_i\}_{i=0}^{t-1} \Big)
\]
\end{algorithmic}
\end{algorithm}

\begin{algorithm}[H]
\caption{\textsc{Eidolon.Verify}$(G,k,M,\sigma)$}
\begin{algorithmic}[1]
\State \textbf{Input:} public parameters \((G,k)\), message \(M\), signature \(\sigma\)
\State Parse \(\sigma\) as \(\{C^{(i)}\}_{i=0}^{t-1}\) and openings \(\{\mathit{open}_i\}_{i=0}^{t-1}\)
\State Recompute
\[
h' \gets \mathcal H\big(\mathrm{Encode}(G,k,C^{(0)},\dots,C^{(t-1)},M)\big)
\]
\State Parse \(h'\) into \(e_0',\dots,e_{t-1}'\), where \(e_i'=(u_i',v_i')\)
\For{$i=0,\dots,t-1$}
    \State Let \(\mathit{open}_i=(\alpha_u,r_u,\alpha_v,r_v)\)
    \State Check that
    \[
    f(\alpha_u,r_u) \stackrel{?}{=} c_{u_i'}^{(i)}
    \qquad\text{and}\qquad
    f(\alpha_v,r_v) \stackrel{?}{=} c_{v_i'}^{(i)}.
    \]
    \If{either equality fails}
        \State \Return Reject
    \EndIf
    \If{\(\alpha_u=\alpha_v\) or \(\alpha_u,\alpha_v\notin\{1,\dots,k\}\)}
        \State \Return Reject
    \EndIf
\EndFor
\State \Return Accept
\end{algorithmic}
\end{algorithm}

\subsection{Merkle compression of commitments}

In the plain scheme, each round \(i\) includes the full commitment vector
\[
C^{(i)}=(c_v^{(i)})_{v\in V},
\qquad
c_v^{(i)} = f\big(\pi_i(\phi(v)),\, r_v^{(i)}\big)\in\{0,1\}^s.
\]
If all \(t\) commitment vectors are included in the signature, the size is
\[
|\sigma|_{\mathrm{plain}}
\;\approx\;
t\,n\,s + 2t\bigl(|\alpha|+|r|\bigr),
\]
where \(n=|V|\), \(s\) is the commitment length, \(|\alpha|\le \lceil\log_2 k\rceil\), and \(|r|\) is the commitment randomness length. This yields asymptotic size \(O(tn)\).

We compress each vector \(C^{(i)}\) to a single Merkle root \(R_i\) using a collision-resistant hash function \(H:\{0,1\}^*\to\{0,1\}^\lambda\). Let \(\mathrm{enc}(v)\) be a fixed-length encoding of the vertex index \(v\), and define
\[
L_v^{(i)} = H\big(\mathsf{leaf}\,\|\,\mathrm{enc}(v)\,\|\,c_v^{(i)}\big).
\]

The hash function \(\mathcal H\) is modeled as a random oracle. For concrete instantiations, one should use a standardized hash with output length and security strength matched to the target security level.

A binary Merkle tree is built over the leaves \((L_v^{(i)})_{v\in V}\), and \(R_i\) denotes its root. Instead of publishing the whole vector \(C^{(i)}\), the signer publishes only \(R_i\). For each challenged vertex, the signature includes the corresponding opening \((\alpha,r)\) together with a Merkle authentication path.

In the Fiat--Shamir transformation, the hash is applied to the roots rather than to the full commitment vectors:
\[
h = \mathcal H\!\big(\mathrm{Encode}(G,k,R_0,\dots,R_{t-1},M)\big).
\]
The challenge edges are then derived from \(h\) as before. For each round \(i\), if the challenge is \(e_i=(u_i,v_i)\), the signer reveals
\[
(\alpha_{u_i},r_{u_i},\mathsf{path}_{u_i}^{(i)})
\qquad\text{and}\qquad
(\alpha_{v_i},r_{v_i},\mathsf{path}_{v_i}^{(i)}),
\]
where \(\alpha_{u_i}=\pi_i(\phi(u_i))\) and \(\alpha_{v_i}=\pi_i(\phi(v_i))\).

\begin{lemma}
Suppose \(H\) is collision-resistant and \(f\) is binding. Fix a round \(i\) and a vertex index \(v\). Given the root \(R_i\), it is infeasible to produce two distinct valid openings for position \(v\).
\end{lemma}

\begin{proof}
If two distinct openings \((\alpha,r)\neq(\alpha',r')\) produce the same commitment value, then the binding property of \(f\) is violated. Otherwise the corresponding leaves differ. If both leaves authenticate to the same root \(R_i\), this yields a collision in the Merkle tree, contradicting the collision resistance of \(H\).
\end{proof}

Each Merkle root contributes \(\lambda\) bits, and each authentication path has length \(\lceil\log_2 n\rceil\). The resulting signature size is therefore
\[
|\sigma|_{\mathrm{Merkle}}
\;\approx\;
t\lambda
+
2t\bigl(|\alpha|+|r|+\lambda\lceil\log_2 n\rceil\bigr),
\]
which gives asymptotic size \(O(t\log n)\).

A Python prototype of the Merkle-compressed construction was implemented to validate the encoding, challenge derivation, and signature-size formulas. On a small prototype instance, the measured signature sizes matched the analytic expressions derived above.

\section{Security model and analysis}

We use the standard notion of existential unforgeability under adaptive chosen-message attack (EUF--CMA) for signature schemes, and we analyze the construction in the random-oracle model.

\begin{proposition}[Basic consistency of a valid forgery]
Assume that \(f\) is computationally binding, that \(H\) is collision-resistant, and that \(\mathcal H\) is modeled as a random oracle. Let
\[
\sigma=\bigl(\{R_i\}_{i=0}^{t-1},\{\mathit{open}_i\}_{i=0}^{t-1}\bigr)
\]
be a signature accepted by the verifier for a message \(M\). Then, except with negligible probability, for every round \(i\), the opening \(\mathit{open}_i\) is consistent with the challenged edge derived from
\[
h=\mathcal H(\mathrm{Encode}(G,k,R_0,\dots,R_{t-1},M)),
\]
and with the commitments authenticated under the Merkle root \(R_i\). In particular, any successful forgery must either:
\begin{enumerate}
\item produce valid openings for all challenged positions under the corresponding Merkle roots, or
\item violate the binding of \(f\) or the collision resistance of \(H\).
\end{enumerate}
\end{proposition}

\begin{proof}
Let
\[
h=\mathcal H(\mathrm{Encode}(G,k,R_0,\dots,R_{t-1},M))
\]
be the hash value recomputed by the verifier, and let
\[
(e_0,\dots,e_{t-1})
\]
be the sequence of challenged edges derived from \(h\). Since the verifier is deterministic once \((G,k,M,\sigma)\) is fixed, acceptance means that, for each round \(i\), the verifier has parsed \(\mathit{open}_i\), extracted the challenged edge \(e_i=(u_i,v_i)\), and accepted all checks attached to that round.

Fix a round \(i\). By construction of the verification algorithm, acceptance implies that the verifier has checked the Merkle authentication paths for the two challenged endpoints \(u_i\) and \(v_i\), starting from the corresponding leaf values and ending at the published root \(R_i\). Therefore, if one of these two openings were not consistent with \(R_i\), then either the verifier would reject, or two distinct leaf values would authenticate to the same root. In the latter case, this would contradict the collision resistance of \(H\).

Next, the verifier also checks that the revealed pairs \((\alpha_u,r_u)\) and \((\alpha_v,r_v)\) open the corresponding commitments. Hence, if an accepted opening at one of the challenged positions admitted two distinct valid decommitments, this would violate the binding property of \(f\).

Therefore, except with negligible probability, every accepted round is simultaneously consistent with the challenged edge derived from the hash value, with the openings checked against the commitment function \(f\), and with the Merkle root \(R_i\). Applying the same argument to all rounds yields the claim.
\end{proof}

The proposition above only establishes consistency of accepted transcripts with the published commitment structure. A complete EUF--CMA reduction would additionally require a witness-extraction argument for the Fiat--Shamir transform in this setting.

\section{Security Analysis and Attack Models}
\label{sec:attacks}

\subsection{Exact attacks}

To set appropriate parameters and evaluate the security of \textbf{Eidolon}, we analyze the practical difficulty of recovering the secret $k$-coloring from the public graph using a well-known algorithmic approach. In the context of this paper, we are interested in any coloring of the random graph with $k$ or less colors, as this allows us to impersonate the prover. For a comprehensive analysis of the problem, see Mann \cite{Mann}. We use the \textbf{exactcolors} algorithm \cite{HCS12}, which is based on branch-and-bound with a linear programming solver (Gurobi). This is currently one of the most efficient solvers for the coloring problem on ER-graphs \cite{BFHM}.

We begin by confirming the properties of the random graph and the difficulty of finding the chromatic number of a random graph. Figure \ref{fig:ERgraphs} illustrates the chromatic numbers found using the \textbf{exactcolors} solver for graphs drawn from $G_R(n,\frac{1}{2})$ with between 10 and 52 vertices and the minimal coloring found using the DSatur heuristic. We choose $p=1/2$ as this maximizes the complexity of the graphs and the associated coloring problem. The Bollabas upper and lower bounds are also plotted. It is clear that the actual chromatic numbers are much larger than the  asymptotic formula, which does not apply in this range. Instead, the chromatic number seems to follow a power law with best-fit equation $\chi(n,1/2)=0.88n^{0.611}$. 
\begin{figure}[H]
    \begin{center}
        \includegraphics[width=1.0\linewidth]{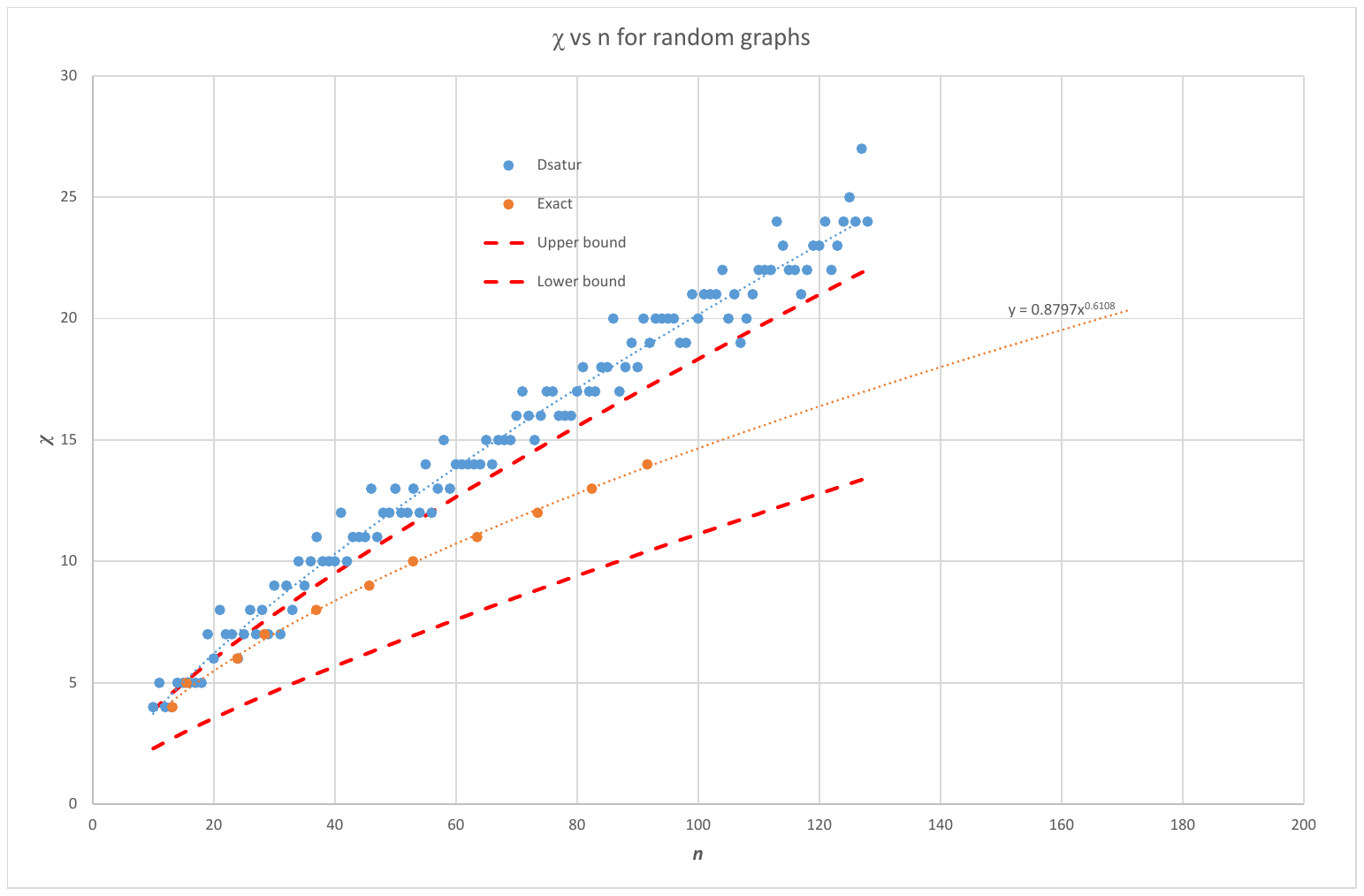}
    \end{center}
    \caption{The observed chromatic number of random graphs using an exact algorithm and the minimal coloring discovered by the DSatur algorithm. The upper and lower bounds of $\chi$ due to Bollob\'as are also plotted.}
    \label{fig:ERgraphs}
\end{figure}

The secret key is embedded as a planted coloring in the graph. It seems natural to choose $\chi$ as the number of colors for the key, but embedding an additional solution may change the difficulty of the problem. To confirm the correct number of colors, we analyzed the solution time to find a coloring of size $k$ embedded in graphs with expected $\chi=13(n=72)$, $\chi=14(n=82)$ and $\chi=15(n=93)$ (Figure \ref{fig:difficulty}. The peak difficulty occurs either at $k=\chi$ or $k=\chi+1$, and so we choose $k=\chi$.

\begin{figure}[H]
    \begin{center}
        \includegraphics[width=1.0\linewidth]{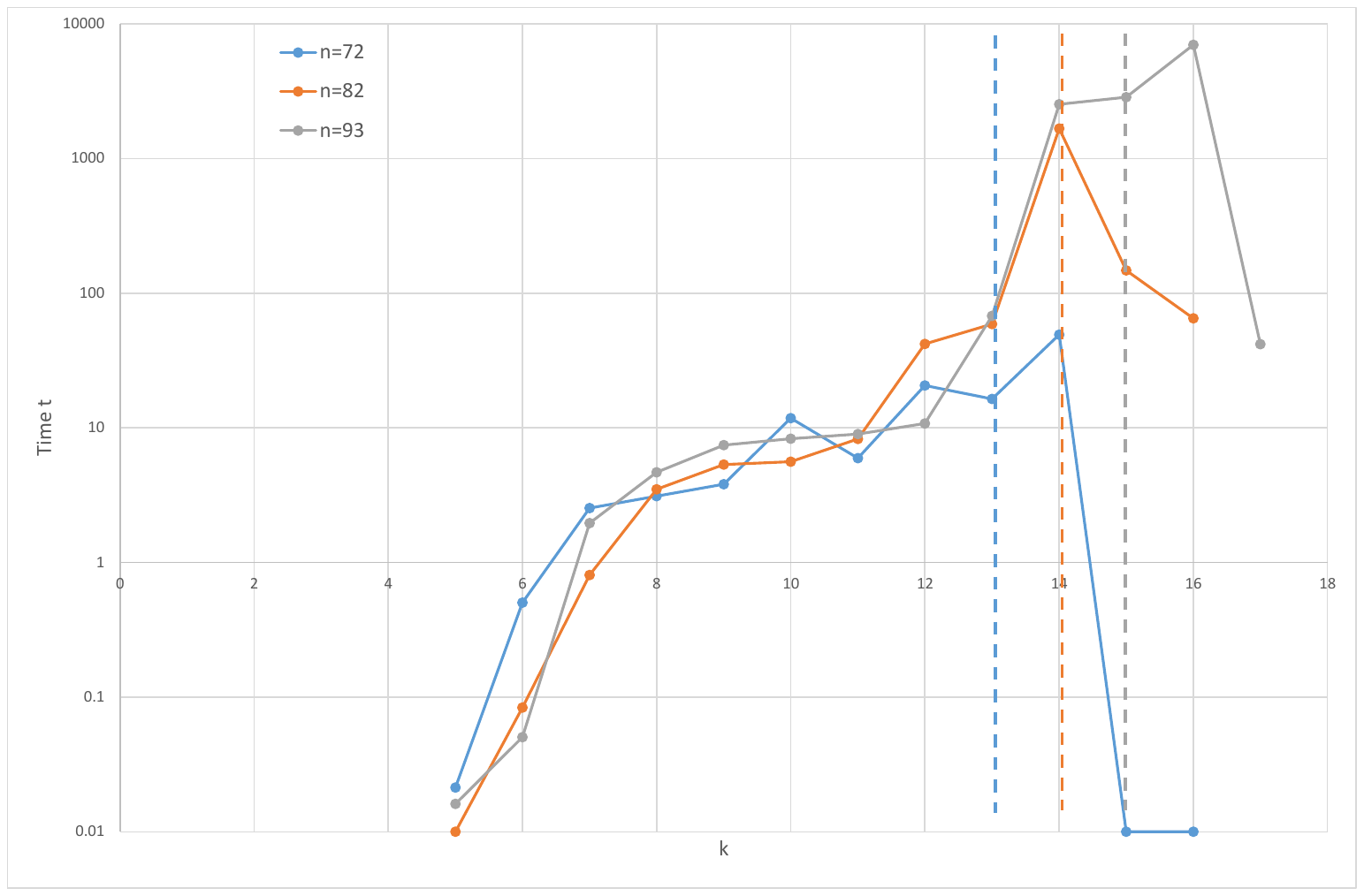}
    \end{center}
    \caption{The time taken to find the minimal coloring of graphs with different planted solution size $k$ for graphs of size 72, 82 and 93. The vertical dashed lines represent the expected chromatic number of the equivalent ER graphs.}
    \label{fig:difficulty}
\end{figure}

%We now look at planted solutions using the algorithm described in Section \ref{PlantedSolution} with $n=36$ and $s=1/2$, with $p$ selected as above. The time to find a solution of the correct size is shown in Figure \ref{fig:PlantedGraphs}. As we might hope, the time-to-solution is exponential in $k$, up to a threshold value, where it becomes constant. At this point, the natural minimum coloring of the equivalent random graph has less colors that the planted solution and is found by the algorithm. The computational complexity is therefore bounded by this solution and constant (for fixed $n,p$). As in the previous examples, we observe that the transition occurs at $k=8$ not $n/2(log_2 n-1)\approx 4$. We should therefore select $k=\chi$ for maximum difficulty of the k-coloring problem.

%\begin{figure}[H]
%    \begin{center}
%        \includegraphics[width=1.0\linewidth]{exactplanted.png}
%    \end{center}
%    \caption{The time taken for our exact solver to find a minimal coloring of the graph with a planted coloring of size $k$. Below $k=8$, the algorithm discovers the planted solution. Above this point, the algorithm finds a natural coloring of the random graph with size around 8.}
%    \label{fig:PlantedGraphs}
%\end{figure}

Finally, we demonstrate the relationship between $n$ and solution time in Figure \ref{fig:n_vs_t}. We select $k$ for the planted solution to match the empirical prediction of $\chi$ for the equivalent random graph. The solution times appears, empirically, to be super-exponential. These results were produced on a AMD EPYC 7501 at 2.6GHz with 512GB of RAM.

\begin{figure}[H]
    \begin{center}
        \includegraphics[width=1.0\linewidth]{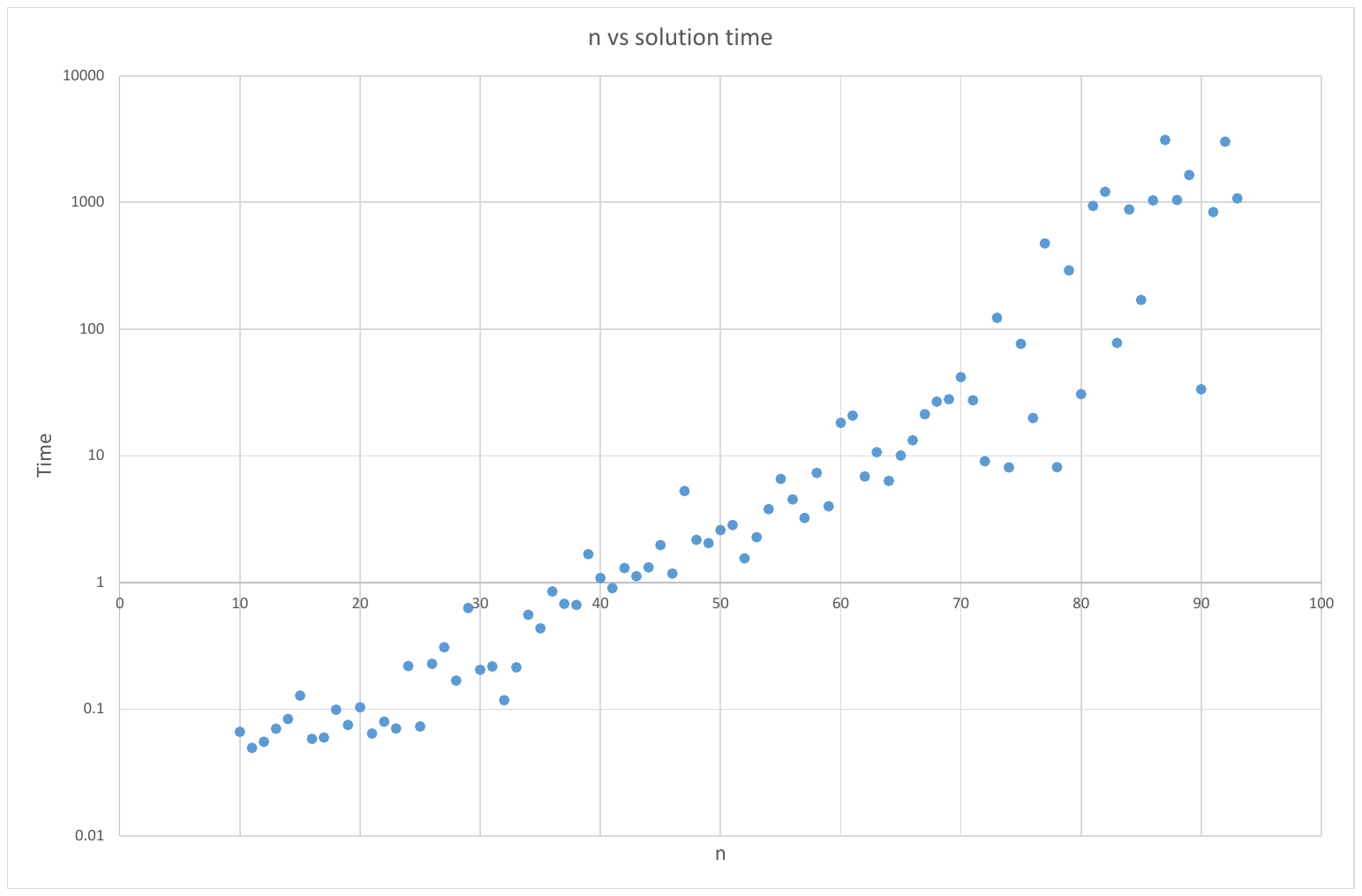}
    \end{center}
    \caption{The time-to-solution for the `exactcolors' solver on planted graphs with planted $k$ chosen according to the empirical power law.}
    \label{fig:n_vs_t}
\end{figure}

\subsection{Heuristic attacks}

The Dsatur algorithm \cite{Brelaz}  is a well-known greedy algorithm for coloring graphs. The algorithm sequentially selects vertices for coloring on the basis of the maximum current degree of saturation. It can recover the minimal coloring for small graphs but is known to be limited for larger graphs. We can see this in in Figure \ref{fig:ERgraphs}, where the size of the solution found by Dsatur is shown, along with the chromatic number discovered by the exact algorithm. While the method performs well for small graphs, a gap is evident for larger ones and DSatur cannot recover the minimal coloring for graphs larger than $n=40$. This result is confirmed in Figure \ref{fig:plantedgraphs} for planted graphs, where we can see identical behavior to the ER-graphs.

In \cite{GHK}, Gryak, Haralick, Kahrobaei have used machine learning algorithms to solve one of the algorithmic problems, known as conjugacy decision problem for certain classes of groups. This technique has been used for the cryptanalysis of proposed schemes using the conjugacy problem.
Machine learning and pattern recognition techniques have been successfully applied to algorithmic problems in free groups. In \cite{GHK}, the authors seek to extend these techniques to finitely presented non-free groups, with a particular emphasis on polycyclic and metabelian groups that are of interest to non-commutative cryptography. 
As a prototypical example, they utilize supervised learning methods to construct classifiers that can solve the conjugacy decision problem, i.e., determine whether or not a pair of elements from a specified group are conjugate. The accuracies of classifiers created using decision trees, random forests, and N-tuple neural network models are evaluated for several non-free groups. The very high accuracy of these classifiers suggests an underlying mathematical relationship with respect to conjugacy in the tested groups.

%\href{https://www.researchgate.net/publication/367084178_Solving_Graph_Coloring_Problem_via_Graph_Neural_Network_GNN}{Solving Graph Coloring Problem via Graph Neural Network (GNN)} to compare with our results.

Based on these results, we propose a modified version of the Graph Neural Network framework to tackle the GCP. The method is based on that of Schuetz et al \cite{Schuetzetal}. The method uses a set of random features initially assigned to the vertices of the graph. The GNN is then used to decode these features into color labels for the vertices. The network learns through a cost function to minimize the number of color conflicts in the assignment.  We use the GraphSAGE variant described in the paper with one hidden layer. We optimized the parameter settings to find $d_0=45, d_1=40,\beta=0.0091$, dropout=0.6, epochs=20000.

We use a modified cost function suggested by Porumbel et al \cite{PorumbelHaoKuntz} which weights color violations by the inverse of the degree, favoring violations with a smaller number of neighbors to resolve. If the softmax output of the network is $\mathbf{p}_k$ for each color $k$, then the standard cost function is given by $\sum_k\mathbf{p}_k^T\mathbf{A}\mathbf{p}_k$. The modified cost function is
\begin{equation}
\sum_k\mathbf{p}_k^T(\mathbf{I}-\mathbf{D}^{-1})\mathbf{A}\mathbf{p}_k
\end{equation}
This down-weights the cost of violations on low degree vertices, as these are potentially easier to resolve.

\footnote{Our GNN-based attack code is available at \url{https://github.com/EsmaMaysan/GNN-CSP/blob/main/GNN+CSP.py}.}

Since the network is fixed, the expected number of colors, $k$, must be chosen in advance. This is not straightforward, as we found that values of $k$ higher than the known chromatic number of the graph produced better solutions in some cases. However, since the goal is to recover a coloring equivalent to the planted coloring (in order to break the security of the key), we use the same $k$ as the planted coloring.

We found considerable improvement from the network by initializing it with the labels given by the DSatur algorithm. This is achieved by pre-training the network with the DSatur labelling as the target output, using 70 epochs, $\beta=0.08$ and a cross-entropy loss. Since the DSatur algorithm will typically produce more than $k$ colors, the coloring is reduced by replacing the excess colors with the least-conflicted of the initial $k$ colors before training.

Finally, the network does not usually converge to a non-conflicted configuration (i.e. a graph coloring). We therefore post-process the labelling in a similar fashion to \cite{Schuetzetal}, by relabelling conflicted nodes with the smallest number of additional colors. 

\begin{figure}[H]
    \begin{center}
        \includegraphics[width=1.0\linewidth]{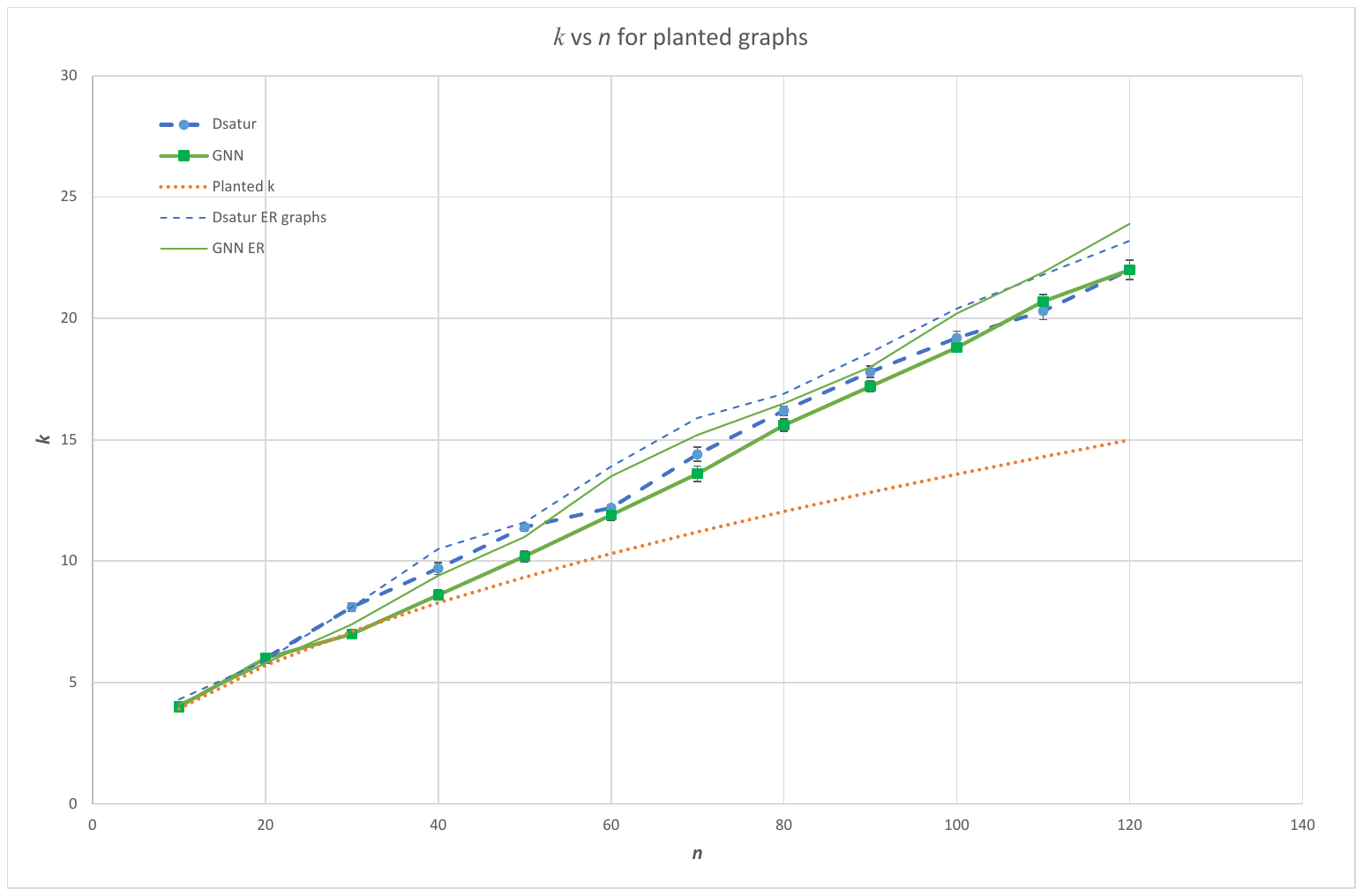}
    \end{center}
    \caption{The size of the recovered colorings for planted graphs using approximate algorithms. The dotted line shows the number of colors $k$ used in the planted graphs. The dashed line gives the number of colors used by DSatur, and the solid line those used by the GNN algorithm. The lighter lines show the results for these algorithms on ER graphs, i.e. random graphs with no planted solution.}
    \label{fig:plantedgraphs}
\end{figure}

The results are shown in Figure \ref{fig:plantedgraphs}. The GNN algorithm marginally outperforms the Dsatur algorithm for intermediate sized graphs ($n\in [20,60]$), but cannot recover the planted coloring for $n\geq 20$. Above $n=60$, the performance is very similar to Dsatur. This suggests for $n>60$ it is very challenging to recover the key with known algorithms. The plot also contains results for ER random graphs (lighter lines) showing very similar performance to the planted graphs.

\section{Conclusion}
We have presented \Eidolon{}, a practical post-quantum digital signature scheme based on the $k$-colorability problem, combining a generalized zero-knowledge protocol, Merkle-tree compression, and carefully constructed hard instances via planted colorings. 

Our empirical security analysis provides a detailed evaluation of the resistance of the scheme against both exact and heuristic attacks. In particular, the results obtained using the \textbf{exactcolors} solver indicate that the time required to recover a valid $k$-coloring grows super-exponentially with the graph size when parameters are chosen according to the empirical chromatic number. This confirms that recovering the planted coloring becomes computationally infeasible even for moderate values of $n$.

On the heuristic side, both the DSatur algorithm and the proposed GNN-based approach fail to recover the planted coloring beyond small graph sizes. While the GNN model shows slight improvements over DSatur for intermediate sizes ($n \in [20,60]$), it does not scale effectively and exhibits similar limitations for larger graphs. In particular, for $n \geq 60$, neither method is able to produce a valid coloring matching the planted solution, and their performance remains close to that observed on random graphs without planted structure.

Overall, these results demonstrate that the hardness of the underlying combinatorial problem is preserved even in the presence of modern heuristic and machine learning techniques. This supports the viability of graph coloring as a foundation for post-quantum cryptographic constructions when parameters are carefully selected.

Future work includes refining parameter selection to align with standardized post-quantum security levels, improving the efficiency of the construction, and further investigating the resilience of the scheme against more advanced attacks, including hybrid and adaptive strategies.

\section*{Acknowledgements}
The authors acknowledge the support from
the Institut Henri Poincaré (UAR 839 CNRS-Sorbonne Université) and LabEx CARMIN (ANR-10-LABX-59-01). This project started from discussions with Farinaz Koushanfar (UCSD) by Delaram Kahrobaei. Consequently restarted during the Trimester on Post-quantum Algebraic Cryptography at IHP. DK conducted this work partially with the support of ONR Grant 62909-24-1-2002. DK thank Institut des Hautes \'Etudes Scientifiques - IHES for providing stimulating environment while this project was partially done.
DK was supported by the CARMIN fellowship during the completion of this project. RF thanks IMUS-Maria de Maeztu grant CEX2024-001517-M - Apoyo a Unidades de Excelencia María de Maeztu for supporting this research, funded by MICIU/AEI/ 10.13039/501100011033".

%\bibliographystyle{amsalpha}
%\nocite{*}
%\bibliography{refs}

\providecommand{\bysame}{\leavevmode\hbox to3em{\hrulefill}\thinspace}
\providecommand{\MR}{\relax\ifhmode\unskip\space\fi MR }
\providecommand{\MRhref}[2]{%
  \href{http://www.ams.org/mathscinet-getitem?mr=#1}{#2}
}
\providecommand{\href}[2]{#2}

\appendix

\section{Standard security definitions}

\begin{definition}[Probabilistic polynomial-time algorithms and negligible functions {\cite[Sec.~1.3]{GoldreichFoC1},\cite[Sec.~3.1]{KatzLindell}}]
An algorithm \(\mathcal{A}\) is \emph{probabilistic polynomial time} (PPT) if there exists a polynomial \(p(\cdot)\) such that, for every input \(x\in\{0,1\}^n\) and every random coin string \(r\), the execution \(\mathcal{A}(x;r)\) halts within at most \(p(n)\) steps. A function \(\mu:\mathbb{N}\to[0,1]\) is \emph{negligible} if, for every polynomial \(q(\cdot)\), there exists \(N\) such that \(\mu(n)<1/q(n)\) for all \(n\ge N\).
\end{definition}

\begin{definition}[EUF--CMA security {\cite[p.~21]{GoldwasserMicaliRivest},\cite[Ch.~12]{KatzLindell}}]
Let \(\Pi=(\mathsf{Gen},\mathsf{Sign},\mathsf{Vfy})\) be a signature scheme with security parameter \(\lambda\). The EUF--CMA experiment is defined as follows.

\begin{enumerate}
\item \emph{Setup:} The challenger samples \((\mathsf{pk},\mathsf{sk})\leftarrow \mathsf{Gen}(1^\lambda)\) and gives \(\mathsf{pk}\) to the adversary \(\mathcal{A}\).
\item \emph{Signing queries:} The adversary \(\mathcal{A}\) is given adaptive oracle access to \(\mathsf{Sign}_{\mathsf{sk}}(\cdot)\). For each queried message \(m\), it receives a signature \(\sigma\leftarrow \mathsf{Sign}_{\mathsf{sk}}(m)\). Let \(Q\) denote the set of queried messages.
\item \emph{Forgery:} Eventually, \(\mathcal{A}\) outputs a pair \((m^\star,\sigma^\star)\).
\end{enumerate}

The adversary wins if
\[
\mathsf{Vfy}_{\mathsf{pk}}(m^\star,\sigma^\star)=1
\qquad\text{and}\qquad
m^\star\notin Q.
\]
The scheme \(\Pi\) is \emph{EUF--CMA secure} if every PPT adversary wins with only negligible probability in \(\lambda\).
\end{definition}

\end{document}